\documentclass[conference,compsocconf]{IEEEtran}
\usepackage[pdftex]{graphicx}
\usepackage[compatible]{algpseudocode}
\usepackage{algorithm}
\usepackage{amsfonts}
\usepackage{amsthm}
\usepackage{url}

\newtheorem{theorem}{Theorem}
\newtheorem{definition}[theorem]{Definition}
\newtheorem{corollary}[theorem]{Corollary}

\usepackage[
 lambda,
 operators,
 advantage,
 sets,
 adversary,
 landau,
 probability,
 notions,
 logic,
 ff,
 mm,
 primitives,
 events,
 complexity,
 asymptotics,
 keys,
]{cryptocode}

\newcommand{\indcpads}{\pcnotionstyle{IND\pcmathhyphen{}CPA\pcmathhyphen{}DS}}
\newcommand{\indcka}{\pcnotionstyle{IND\pcmathhyphen{}CKA}}
\newcommand{\indocpa}{\pcnotionstyle{IND\pcmathhyphen{}OCPA}}
\newcommand{\Leak}{\mathcal{L}}
\newcommand{\Mms}{\mathbb{M}}
\newcommand{\Cds}{\mathbb{C}}
\newcommand{\eseds}{\pcalgostyle{ESEDS}}
\newcommand{\seds}{\pcalgostyle{SEDS}}
\newcommand{\eds}{\pcalgostyle{EDS}}
\newcommand{\ds}{\pcalgostyle{DS}}
\newcommand{\pe}{\pcalgostyle{PE}}
\newcommand{\pse}{\pcalgostyle{PSE}}
\newcommand{\ppke}{\pcalgostyle{PPKE}}
\newcommand{\der}{\pcalgostyle{KDer}}
\newcommand{\search}{\pcalgostyle{Search}}

\renewcommand{\state}{\pckeystyle{st}}

\newcommand{\intersentencespace}{\spacefactor\sfcode`. \space}

\begin{document}

\date{}

\title{An Efficiently Searchable Encrypted Data Structure for Range Queries}

\author{
\IEEEauthorblockN{Florian Kerschbaum}
\IEEEauthorblockA{University of Waterloo\\
Waterloo, Ontario, Canada\\
Email: florian.kerschbaum@uwaterloo.ca}
\and
\IEEEauthorblockN{Anselme Tueno}
\IEEEauthorblockA{SAP \\
Karlsruhe, Germany \\
Email: anselme.kemgne.tueno@sap.com}
}

\maketitle


\begin{abstract}
At CCS 2015 Naveed et al.~presented first attacks on efficiently searchable encryption, such as deterministic and order-preserving encryption.
These plaintext guessing attacks have been further improved in subsequent work, e.g.~by Grubbs et al.~in 2016.
Such cryptanalysis is crucially important to sharpen our understanding of the implications of security models.
In this paper we present an efficiently searchable, encrypted data structure that is provably secure against these and even more powerful chosen plaintext attacks.
Our data structure supports logarithmic-time search with linear space complexity.
The indices of our data structure can be used to search by standard comparisons and hence allow easy retrofitting to existing database management systems.
We implemented our scheme and show that its search time overhead is only $10$ milliseconds compared to non-secure search.
\end{abstract}




\newlength{\myimagesize}
\setlength{\myimagesize}{\columnwidth}

\section{Introduction}

At CCS 2015 Naveed et al.~\cite{NavKam15} presented attacks on order-preserving encryption.
Later Grubbs et al.~\cite{GruSek16} improved the precision of these attacks.
Further attacks on searchable encryption have been presented \cite{CasGru15,DurDuB16,GruMcP16,IslKuz12,KelKol16,LacMin17,PouWri16,ZhaKat16}.
Such cryptanalysis is crucially important to sharpen our understanding of the implications of security models, since many of the attacked encryption schemes are proven secure in their specific security models.
In this paper we formalize security against these attacks and show a connection to chosen plaintext attacks.
We also demonstrate that there exists an encrypted data structure that supports efficient range queries by regular comparisons and that provably prevents these attacks.

\begin{table*}[!t]
\centering
\caption{Comparison of efficiency and security of schemes for range queries over encrypted data}
\label{tbl:comparison}
\begin{tabular}{l||c|c|c}
Scheme                                                       & Search Time $O(\log n)$ & Space $O(n)$ & $\indcpads$-secure\footnotemark \\
\hline
\hline
Partial order preserving encoding \cite{RocApo16}            & Yes                     & Yes          & Only before queries \\
\hline
Searchable encryption with replicated index \cite{DemPap16}  & Yes                     & No           & Yes \\
\hline
Searchable encryption with dynamic index \cite{HahKer16}     & Only amortized          & Yes          & Only before queries \\
\hline
Searchable encryption with index replacement \cite{BoePod16} & Yes                     & Yes          & No \\
\hline
Order-revealing encryption \cite{LewWu16}                    & No                      & Yes          & Yes \\
\hline
This paper                                                   & Yes                     & Yes          & Yes \\
\end{tabular}
\end{table*}

Comparison using regular comparison operators (e.g.~greater-than) as enabled by our scheme and order-preserving encryption has many practical benefits.
These encryption schemes can be retrofitted to any existing database management system making them extra-ordinarily fast, flexible and easy-to-deploy.
We preserve this property as our implementation demonstrates, but some minor modifications to the search procedure are necessary.

Efficiency -- logarithmic time and linear space complexity -- is also an important property of search over encrypted data.
In Table \ref{tbl:comparison} we provide a comparison of our scheme to the most secure and efficient order-preserving schemes \cite{RocApo16}, order-revealing encryption \cite{LewWu16} and range-searchable encryption \cite{BoePod16,DemPap16,HahKer16} schemes.
No searchable encryption scheme -- including ours -- offers perfect security and efficiency for all functions (equality and range search, insertions, deletions, etc.).
It is a research challenge to balance the trade-off between the two objectives, even for a restricted set of functions.
We aim at provable security against the recently publicized plaintext guessing attacks while still enabling efficient range search.
In this respect, we achieve a novel and preferable trade-off between security and efficiency.

In the construction of our scheme we borrow the ideas of previous order-preserving encryption schemes: modular order-preserving encryption by Boldyreva et al.~\cite{BolChe11}, ideal secure order-preserving encoding by Popa et al.~\cite{PopLi13} and frequency-hiding order-preserving encryption by Kerschbaum~\cite{Ker15}.
We assign a distinct ciphertext for each -- even repeated -- plaintext as Kerschbaum does, but his scheme statically leaks the partial insertion order.
So, we compress the randomized ciphertexts to the minimal ciphertext space using Popa et al.'s interactive protocol.
Then we rotate around a modulus as Boldyreva et al., but on the ciphertexts and not on the plaintexts.

As a result we achieve {\em structural independence} between the ciphertexts and plaintexts which is a prerequisite for security against chosen plaintext attacks and plaintext guessing attacks -- particularly, if the adversary has perfect background knowledge on the distribution of plaintexts.
We formalize this insight as a novel security model ($\indcpads$-security) for efficiently searchable, encrypted data structures and we prove our scheme secure in this model.
Our security model encompasses a number of recently publicized attacks where attackers broke into cloud system and stole the stored data.
Such an attack will reveal no additional information when data is encrypted with our scheme.
This will also thwart the attacks by Naveed et al.~\cite{NavKam15} and Grubbs et al.~\cite{GruSek16} mentioned at the beginning of the introduction. 

The implementation of our scheme shows only $10$ milliseconds overhead compared to non-secure search on a database with a million entries.
In summary our contribution are as follows:

\begin{itemize}

\item We formulate a security notion that provably prevents chosen plaintext attacks and plaintext-guessing attacks as those by Naveed et al.~and Grubbs et al.
Our model provides provable security against attackers with (one-time) snapshot access to the encrypted data as in the most common attacks on cloud computing.

\item We present an efficiently searchable, encrypted data structure that supports range queries and fulfills this security notion.
Our search scheme is retrofittable into existing database management systems and we provide a prototypical implementation.

\item We evaluate the performance of our scheme in a prototypical implementation.
Our scheme shows only roughly $10$ milliseconds overhead compared to non-secure search.

\end{itemize}

The remainder of the paper is structured as follows.
In the next section we define what we mean by an efficiently searchable, encrypted data structure.
In Section \ref{sec:security} we present and motivate our new security model preventing plaintext guessing attacks.
Then, we present our efficiently searchable, encrypted data structure secure in this model in Section~\ref{sec:scheme}.
We evaluate the performance of the implementation of our scheme in \ref{sec:evaluation}.
Finally, we review related work in Section~\ref{sec:related} and summarize our conclusions in Section~\ref{sec:conclusions}.

\section{Efficiently Searchable En\-cryp\-ted Data Structures}
\label{sec:edese}

First, we define what we mean by a efficiently searchable encrypted data structure ($\eseds$).
We start by defining what we mean by a data structure.
We use the fundamental representation of a data structure in random-access memory, i.e.~an array.
Each cell of the array consists of a structured element.
We do not impose any restriction on the structure of the element, but usually this element contains two parts: the data to be searched over and further structural information, such as indices of further entries.
Note that structural information may be {\em implicit}, i.e.~the index where an element is stored itself is structural information albeit not explicitly stored.
This implicit structural information may also not be encrypted, but only randomized.
An example of {\em explicit} structural information are the indices of the cells of the two children in a binary search tree which would be stored in a cell's structure in addition to the data of a tree node.
Explicit structural information can be encrypted.
We write $\Cds[j]$ for the $j$-th element and if it is clear from the context, we assume it consists only of a ciphertext of the data element (with $j$ being the implicit structural information).

\begin{definition}[\ds]
\label{defn:ds}

A data structure $\ds$ consists of an array of elements $\Cds[j]$ ($0 \leq j < n$).

\end{definition}

For an encrypted data structure there are a number of options on the type of encryption.
First, we can choose symmetric or public-key encryption.
We can instantiate our encrypted data structure with either one.
Let $\pse$ be a probabilistic symmetric encryption scheme consisting of three -- possibly probabilistic -- polynomial-time algorithms $\pse = \kgen(\secparam)$, $\enc(\key, m)$, $\dec(\key, c)$.
Let $\ppke$ be a probabilistic public-key encryption scheme consisting of three -- possibly probabilistic -- polynomial-time algorithms $\ppke = \kgen(\secparam)$, $\enc(\pk, m)$, $\dec(\sk, c)$.
Let $\pk \gets \der(\sk)$ be a deterministic algorithm that derives the public key from the private key in a public-key encryption scheme.
For symmetric key encryption let $\der$ be the identity function.
Let $\pe \in \{ \pse, \ppke \}$ and we use $\pe$ when we leave the choice of encryption scheme open.

\footnotetext{$\indcpads$ is defined in Section \ref{sec:indcpads}.}
Second, we can either encrypt the data structure as a whole or parts of the data structure -- ideally each cell.
Our requirement of efficient search rules out the first option.
Since in this case each search operation would require decrypting the data structure which is at least linear in the ciphertext size, sublinear search is impossible.
Hence, we require each cell to be encrypted as a separate ciphertext.\footnote{In case several cells of a simple data structure are encrypted as a whole, we call this combination a cell of another data structure.}

Third, for data security it may only be necessary to encrypt the data elements of a cell and not the structural information.
In fact, our own proposed $\eseds$ is an instance of such a case where the structural information is implicit from the array structure and unencrypted.
Hence, we only require the data element of each cell to be encrypted.

\begin{definition}[\eds]
\label{defn:eds}

An encrypted data structure $\eds_{\pe}$ consists of an array of elements $\Cds[j]$ where at least the data part has been encrypted with $\pe$.

\end{definition}

We can now define the operations on a searchable encrypted data structure $\seds_{\pe}$.
We write $\seds$ when the choice encryption scheme is clear from the context.
Furthermore we denote sometimes denote the version $h$ (after $h$ insertions) of a data structure as $\seds^h$.
Our definition is for range-searchable encrypted data structures, but this implies a definition for keyword searchable data structure as well (where the range parameters are equal: $a=b$).
Furthermore, we do not define how operations on our $\seds$ are to be implemented.
These operations can be implemented as algorithms running on a single machine or protocols distributed over a client and server (hiding the secret key from the server).
Both choices are covered by our definition.

\begin{definition}[\seds]
\label{defn:seds}

A searchable encrypted data structure $\seds_{\pe}$ offers the following operations.

\begin{itemize}

\item $\key \gets \kgen(\secparam)$: Generates a -- either secret or private -- key $\key$ from the encryption scheme $\pe$ according to the security parameter $\secpar$.

\item $\Cds^{h+1} \gets \enc(\key, \Cds^h, m)$:  Encrypts the plaintext $m$ using $\pe.\enc(\der(\key), m)$ and inserts it into the data structure $\Cds^h$ resulting in data structure $\Cds^{h+1}$.\footnote{Note that in case of public key encryption our definition does not imply that the entire operation can be completed using only the public key.}

\item $m := \dec(\key, \Cds[j]$): Computes the plaintext $m$ for the data part of encrypted cell $\Cds[j]$ using key $\key$.

\item $\{ j_0, \ldots, j_{\ell-1} \} := \search(\key, \Cds, a, b)$:  Computes the set of indices $\{ j_0, \ldots, j_{\ell-1} \}$ for the range $[a, b]$ on the encrypted data structure $\Cds$ using key $\key$.

\end{itemize}

\end{definition}

For the correctness of encryption we expect in a sequence of operations $\enc(\key, \Cds^0, m_0), \ldots, \enc(\key, \Cds^{n-1}, m_{n-1})$ resulting in data structure $\Cds^n$ that $\forall i \, \exists j \, m_i = \dec(\key, \Cds^n[j])$.
For the correctness of search we expect that for any $\{ j_0, \ldots, j_{\ell-1} \} := \search(\key, \Cds, a, b)$, it holds that $\forall j \, \in \{ j_0, \ldots, j_{\ell-1} \} \Longrightarrow \dec(\key, \Cds[j]) \in [ a, b ]$ and 
$\forall j \, \in \{ j | \dec(\key, \Cds[j]) \in [ a, b ] \} \Longrightarrow j \in \{ j_0, \ldots, j_{\ell-1} \}$.

We can now finally define an efficiently searchable encrypted data structure.

\begin{definition}[\eseds]
\label{defn:eseds}

An efficiently searchable encrypted data structure $\eseds$ is a searchable encrypted data structure where the running time $\tau$ of $\search$ is poly-logarithmic in $n$ (plus the size of the returned set of matching ciphertext indices) and the space $\sigma$ of $\eseds$ is linear in $n$:
\begin{align*}
\tau(\search)  & \leq \bigO{\operatorname{polylog}(n) + \ell} \\
\sigma(\eseds) & = \bigO{n}
\end{align*}

\end{definition}

It is now clear that efficient search prevents encrypting the entire data structure and thereby achieving semantic ($\indcpa$) security.
Next, we give our definition of security that implies that each cell's data is encrypted with a semantically secure encryption scheme.
Our security definition also prevents all plaintext guessing attacks of the type of Naveed et al.~and Grubbs et al.
Furthermore, we show that even when the data structure consists of only one semantically secure ciphertext in each cell, this does not guarantee security against these plaintext guessing attacks.

\section{Security of $\eseds$}
\label{sec:security}

Before we define the security of an $\eseds$ we will review recent attacks on cloud infrastructures and searchable encryption scheme to motivate our security model.
Particular we review in depth plaintext guessing attacks that only need a (multi-)set of ciphertexts as input (and do not perform active attacks during encryption or search operations).
We try to generalize these attacks and show that even if all elements in an $\eseds$ are semantically secure encrypted, this does not imply that these attacks are infeasible.

\subsection{Motivation}

Our model is motivated by recent attacks on cloud infrastructures and order-preserving or deterministic encryption.
Not only the theoretic demonstrations, but also real world incidents show the risks of deterministic -- not even order-preserving -- encryption.
In at least one case passwords were encrypted using a deterministic algorithm and many subsequently broken \cite{Duc13}.
The cryptanalysis was performed on stolen ciphertexts only (using additional plaintext hints).
Many other hacking incidents have been recently publicized, e.g.~\cite{Fit14,McC16}, that resulted in leakage of sensitive information -- not necessarily ciphertexts.

All these attacks share a common ``anatomy''.
The hackers are capable to break in, access and copy sensitive information.
They used the opportunity of access to gain as much data as possible in a short time, i.e.~the adversary obtains a {\em static} snapshot.
Note that this does not rule out the revelation of more sophisticated, longitudinal attacks in the future, but underpins the pressure to secure our current systems.

In this respect our model achieves the following:
An attacker gaining access to all ciphertexts stored in an encrypted database does not gain additional information to his background knowledge.
We assume even perfect background knowledge, i.e.~the adversary has chosen all plaintexts.
This may sound contradictory at first -- why would someone break into a database which data he knows.
However, if we are able to show security against such strong adversaries, security holds even if the adversary has less, e.g.~imperfect, background knowledge.

\subsection{Transformation to $\eseds$}

Most commonly plaintext guessing attacks are performed on multi-sets of deterministic or order-preserving ciphertexts which do not impose an order as our $\eseds$ do.
However, there exists a natural connection between these encryption schemes and $\eseds$.
We present transformations that turn deterministic or order-preserving encryption schemes into an $\eseds$ as in Definition \ref{defn:eseds} with equivalent leakage.
Any attack successful on these encryption schemes will be successful on the corresponding $\eseds$.

Our transformation for deterministic encryption is loosely based on the data structure in \cite{CurGar11}.
Let $\Mms$ be the multi-set of plaintexts and $\tilde{\Mms}$ be the set of distinct plaintexts.
We denote the size of a multi-set $\mathbb{M}$ as $|\mathbb{M}|$ and the number of occurrences of element $m$ in multi-set $\mathbb{M}$ as $\#_{\mathbb{M}}m$.
Let $\tilde{m}_i$ be $i$-th distinct plaintext and hence $\#_{\Mms}\tilde{m}_i$ be the number of elements $\tilde{m}_i$ in $\Mms$.
We also denote these elements as $m_{i, h}$ for $h = 0, \ldots, \#_{\Mms}\tilde{m}_i - 1$.
Let $p_{i,h}$ be the index of $m_{i, h}$ in $\eseds$ and $p_{i,h} = -1$, if $h \geq \#_{\Mms}\tilde{m}_i$.
Since deterministic encryption can be stored in a relational table, we use the row identifiers $id_{i,h}$ of each ciphertext as the document identifiers and the data $\tilde{m}_i$ as the keywords.
Let $\prf_\key$ be a keyed, pseudo-random function that maps the domain of keywords onto the size $n$ of the $\eseds$.\footnote{For ease of exposition we assume no collisions.}
Then 
$$
\Cds[\prf_\key(\tilde{m}_i)] \gets \Angle{\pe.\enc(\key, \tilde{m}_i), \pe.\enc(\key, id_{i,0}), p_{i, 1}}
$$
For each data $m_{i, h}$ where $0 < h < \#_{\Mms}\tilde{m}_i$ we store
$$
\Cds[p_{i,h}] \gets \Angle{\pe.\enc(\key, \tilde{m}_i), \pe.\enc(\key, id_{i,h}), p_{i, h+1}}
$$
One reveals $\prf_\key(\tilde{m}_i)$, accesses the corresponding bucket (cell) in the data structure and then traverses the list for efficient (keyword) search.

For deterministic order-preserving encryption we can use a similar transformation as above, but use the order $\operatorname{order}(\tilde{m}_i)$ of the plaintext as the element index.
Instead of hashing the keyword into a bucket, one can use binary search for efficient search on this $\eseds$.
$$
\Cds[\operatorname{order}(\tilde{m}_i)] \gets \Angle{\pe.\enc(\key, \tilde{m}_i), \pe.\enc(\key, id_{i,0}), p_{i, 1}}
$$
In frequency-hiding order-preserving encryption, we no longer have a list of identical ciphertexts, but each ciphertext is unique.
Then we can use the randomized order $\operatorname{rand-order}(m_{i,h})$ where elements are sorted, but ties are broken based on the outcome of a coin flip as defined in \cite{Ker15} as the element index.
However, we no longer need to store the row-identifier, since each ciphertext is unique and can be found in the relational table.
$$
\Cds[\operatorname{rand-order}(m_{i,h})] \gets \pe.\enc(\key, \tilde{m}_i)
$$

These $\eseds$ are susceptible to the same plaintext guessing attacks as those by Naveed et al.~\cite{NavKam15} and Grubbs et al.~\cite{GruSek16} on the respective encryption schemes.
We next review these attacks on these encryption schemes.

\subsection{Plaintext Guessing Attacks}
\label{sec:background}

Naveed et al.~\cite{NavKam15} present a series of attacks on deterministic and order-preserving encryption.
They attack deterministic order-preserving encryption by Boldyreva et al.~\cite{BolChe09}.
Grubbs et al.~improved the precision of the attacks and also extended them to other order-preserving and order-revealing encryption.
Their new attacks are not fundamentally different, but improve the matching algorithm between the assumed and measured frequency. 
However, Grubbs et al.~present the  first attack on frequency-hiding order-preserving encryption (FH-OPE) -- the ``bucketing'' attack.




Let $\Cds$ be the multi-set (a multi-set may potentially include duplicate values) of ciphertexts and $\Mms$ be the multi-set of plaintexts in the background knowledge of the adversary.
We assume that the sizes of the multi-sets are equal: $n = |\Cds| = |\Mms|$.

\subsubsection{Frequency Analysis}
\label{sec:fa}

The frequency analysis attack first computes the histograms $\operatorname{Hist}(\Cds)$ and $\operatorname{Hist}(\Mms)$ of the two multi-sets.
Then it sorts the two histograms in descending order: $\vec{c} := \operatorname{Sort}(\operatorname{Hist}(\Cds))$ and $\vec{m} := \operatorname{Sort}(\operatorname{Hist}(\Mms))$.
The cryptanalysis for $c_i$ is $m_i$, i.e.~the two vectors are aligned.

Naveed et al.~implement the frequency analysis as the $l_P$-opti\-mization attack.
Lacharite and Paterson show that frequency analysis is expected to be the optimal cryptanalysis \cite{LacPat15}, but also that $l_P$-opti\-mization is expected to be close to this optimimum.

In the $l_P$-opti\-mization attack the two histograms are not simply sorted and aligned, but a global minimization is run to find an alignment.
Let $\mathbb{X}$ be the set of $n \times n$ permutation matrices.
The attack then finds $X \in \mathbb{X}$, such that the distance $|| \vec c - X \vec m ||_P$ is minimized under the $l_P$ distance.
For many distances $l_P$ the computation of $X$ can be efficiently (polynomial in $n$) performed using an optimization algorithm, such as linear programming.
The cryptanalysis for $c_i$ is $X[m]_i$, i.e.~the two vectors are aligned after permutation.

The attack works not only for order-preserving encryption, but also for deterministic encryption.
The attack is very successful in experimentally recovering hospital data -- even for such deterministic encryption.
Naveed et al.~report an accuracy of $100\%$ for $100\%$ and $95\%$ of the hospitals for the binary attributes of ``mortality'' (whether a patient has died) and ``sex'', respectively, under deterministic encryption.


\subsubsection{Sorting Attack}

Let $\mathbb{D}$ be the domain of all plaintexts in multi-set $\Mms$.
Let $N = |\mathbb{D}|$ be the size of the domain $\mathbb{D}$.
The sorting attack assumes that $\Cds$ is dense, i.e.~contains a ciphertext $c$ for each $m \in \mathbb{D}$.
The adversary computes the unique elements $\operatorname{Unique}(\Cds)$ and sorts them $\vec c := \operatorname{Sort}(\operatorname{Unique}(\Cds))$ and the domain $\vec d := \operatorname{Sort}(\mathbb{D})$.
The cryptanalysis for $c_i$ is $d_i$, i.e.~the order of the ciphertext and the plaintext are matched.

The attack is $100\%$ accurate, if the ciphertext multi-set is dense.
This is a strong assumption, but already Naveed et al.~present a refinement that works also for low-density data.
This cumulative attack combines the $l_P$-optimi\-zation and sorting attack.
The adversary first computes the histograms $\vec c_1 := \operatorname{Hist}(\Cds)$, $\vec m_1 := \operatorname{Hist}(\Mms)$ and the cumulative density functions $\vec c_2 := \operatorname{CDF}(\Cds)$, $\vec m_2 := \operatorname{CDF}(\Mms)$ of the cipher- and plaintexts.
The attack then finds the permutation $X \in \mathbb{X}$, such that the sum of the distances between the histograms and cumulative density functions $|| \vec c_1 - X \vec m_1 ||_P + || \vec c_2 - X \vec m_2 ||_P$ is minimized.
Again, this can be done using efficient optimization algorithms.
The cryptanalysis for $c_i$ is $X[m]_i$.

The attack is very accurate against deterministic order-preserving encryption as demonstrated by Naveed et al.
They report an accuracy of $99\%$ for $83\%$ of the large hospitals for the attributes of ``age''.
The age column is certainly not low-density, but also not dense (as the success rate shows).

Grubbs et al.~further improve the algorithms in this attack by using bipartite matching.
They report an accuracy for their improved attacks of up to $99\%$ on first names and up to $97\%$ for last names which have much more entropy than age.

\subsubsection{Bucketing Attack on FH-OPE}
\label{sec:bucket-attack}

The extension of the sorting attack -- the bucketing attack -- on FH-OPE proceeds as follows.
The adversary sorts the multi-sets $\vec c := \operatorname{Sort}(\Cds)$ and $\vec m := \operatorname{Sort}(\Mms)$, i.e.~it is not necessary to only use unique values.
The cryptanalysis for $c_i$ is $m_i$.
Note that in FH-OPE every element $c_i \in \Cds$ is unique, but after the attack aligned to the cumulative density function of $\Mms$ as in the cumulative attack.

Grubbs et al.~recover 30\% of first names and 6\% of last names in their data set.
However, it can be even more accurate depending on the precision of the background know\-ledge $\Mms$.
It can be very dangerous to make assumptions about the adversary's background knowledge, since they are hard, if not impossible, to verify and uphold.
Hence, in our $\indcpads$-security model for $\eseds$ we assume {\em perfect} background knowledge of the adversary, i.e.~the multi-set $\Mms$ is the exact same multi-set as the plaintexts of the ciphertexts.
In fact, the multi-set $\Mms$ is chosen by the adversary in a chosen plaintext attack.
The bucketing attack then succeeds with $100\%$ accuracy.

The focus of this paper is preventing these plaintext guessing attacks on efficiently searchable encryption.
However, attacks using stronger adversaries, e.g.~active modifications or insertions, have been presented in the scientific literature \cite{CasGru15,DurDuB16,GruMcP16,KelKol16,LacMin17,PouWri16,ZhaKat16}.

\subsection{Security Definition}
\label{sec:indcpads}

We give our security definition as an adaptation of semantic security to data structures.
We show that our adaptation implies that each data value is semantically secure encrypted.
However, we also show that even if all cells consist of only one semantically secure ciphertext, our adaptation is not necessarily fulfilled.

First, recall the definition of semantic security.

\begin{definition}[IND-CPA]
\label{def:ind-cpa}
A public-key encryption scheme $\ppke$ has \emph{indistinguishable encryptions under a chosen-plaintext attack}, or is \emph{\indcpa-secure}, if for all PPT adversaries $\adv$ there is a negligible function $\negl$ such that

\begin{align*}
 \advantage{\indcpa}{\adv,\ppke} &:=
 \left|
  \prob{\mathrm{Exp}_{\adv,\ppke}^{\indcpa}(\secpar) = 1} - \dfrac{1}{2}
 \right| \\
 & \leq \negl
\end{align*}

\begin{pcvstack}[center]
\procedure
 {$\mathrm{Exp}_{\adv,\ppke}^{\indcpa}(\secpar)$}{
 \Angle{\pk, \sk} \gets \ppke.\kgen(\secparam) \\
 \Angle{m_0, m_1, \state} \gets \adv(\secparam, \pk) \\
  b \sample \bin \\
  c \gets \ppke.\enc(\pk, m_b) \\
  b' \gets \adv(\secparam, \pk, c, \state) \\
  \pcreturn b = b'
}
\end{pcvstack}
\end{definition}

We note that $\indcpa$-security only considers a single ciphertext whereas a data structure consists of multiple ciphertexts and hence some structural information.
Exactly this structural information can be used in plaintext guessing attacks and we need to adapt semantic security to all ciphertexts.
We call our adaptation {\em indistinguishability under chosen-plaintext attacks for data structures} or \indcpads-security for short.
Loosely speaking, our security model ensures that an adversary who has chosen all plaintexts encrypted in a data structure cannot guess the plaintext of {\em any} ciphertext better than a random guess.
Recall that we denote the size of multi-set $\mathbb{M}$ as $|\mathbb{M}|$ and the number of occurrences of element $m$ in multi-set $\mathbb{M}$ as $\#_{\mathbb{M}}m$.

\begin{definition}[IND-CPA-DS]
\label{def:ind-cpa-ds}

An efficiently searchable encrypted data structure $\eseds$ is {\em indistinguishable under a chosen-plaintext attack}, or is {\em \indcpads-secure}, if for all PPT adversaries $\adv$ there is a negligible function $\negl$ such that

\begin{align*}
 \advantage{\indcpads}{\adv,\eseds} &:=
 \left|
  \prob{\mathrm{Exp}_{\adv,\eseds}^{\indcpads}(\secpar) = \Angle{1, p}} - p
 \right| \\
 & \leq \negl
\end{align*}

\begin{pcvstack}[center]
\procedure
{$\mathrm{Exp}_{\adv,\eseds}^{\indcpads}(\secpar)$}{
  \Angle{\pk, \sk} \gets \eseds.\kgen(\secparam) \\
  \Angle{\Mms_0, \Mms_1, \state} \gets \adv(\secparam, \pk) \\
  \pcif |\Mms_0| \neq |\Mms_1| \pcthen \ \pcreturn \bot \\
  b \sample \bin \\
  \Cds := \emptystring \\
  \pcforeach m \in \Mms_b \pcdo \\
  \pcind \Cds \gets \eseds.\enc(\sk, m, \Cds) \\
  \pcendforeach \\
  \Angle{j', m'} \gets \adv(\secparam, \pk, \Cds, \state) \\
  \pcreturn \Angle{\eseds.\dec(\sk, \Cds[j']) = m', \dfrac{\#_{\Mms_0 \cup \Mms_1}m'}{|\Mms_0 \cup \Mms_1|}}
}
\end{pcvstack}
\end{definition}

There are two differences between $\indcpa$-security and $\indcpads$-security.
First, the adversary chooses two multi-sets of plaintexts as input to the challenge instead of two single plaintexts.
This enables the adversary to create different situations to distinguish.
Assume the adversary returns two disjoint multi-sets as $\Mms_0$ and $\Mms_1$, e.g. $\Mms_0 = \{0, 0\}$ and $\Mms_1 = \{1,1\}$.
Then it can attempt to distinguish which of the two plaintext multi-sets have been encrypted by guessing any plaintext in the data structure.
Assume the adversary returns the same multi-set as $\Mms_0$ and $\Mms_1$, but with distinct plaintexts in the (identical) multi-set, e.g. $\Mms_0 = \Mms_1 = \{0,1\}$. 
This is admissible in the definition of $\indcpads$-security, since the only requirement is that the two multi-sets are of the same size.
The adversary can then attempt to distinguish at which position in the data structure each plaintext has been encrypted.

In order to enable the adversary to win the game when the position in the data structure is not indistinguishable, we made a second change to $\indcpa$-security:
The adversary's guess is the plaintext of a single ciphertext at any position in the data structure.
Hence, the adversary does not necessarily have to distinguish between the two plaintext multi-sets, it is sufficient, if it guesses correctly within the choice of sets (which may be equal).
However, even if the position in the ciphertext is indistinguishable, in order to win the adversary only has to guess correctly with a probability non-negligibly better than the frequency of the plaintext in the {\em union of the multi-sets}.
Hence, if the two multi-sets are not equal and the adversary can guess the chosen multi-set, it can win the game.

We next explain the implications of $\indcpads$-security and first prove that $\indcpads$-security implies $\indcpa$-security.
Our proof assumes the use of public-key encryption, but the proof for symmetric encryption is analogous using an encryption oracle.
We prove this by turning an adversary $\bdv$ that has advantage $\epsilon$ in experiment $\mathrm{Exp}_{\bdv,\ppke}^{\indcpa}$ into an adversary $\adv$ that has advantage $\epsilon$ in experiment $\mathrm{Exp}_{\adv,\eseds_{\ppke}}^{\indcpads}$.

\begin{theorem}
\label{thm:ind-cpa}
If $\eseds_{\ppke}$ is $\indcpads$-secure, then each ciphertext of the data element in $\Cds[j]$ $(0 \leq j < n)$ must be from a $\indcpa$-secure encryption scheme.
\end{theorem}

\begin{proof}
Let $\bdv$ be an adversary that has advantage $\epsilon$ in experiment $\mathrm{Exp}_{\bdv,\ppke}^{\indcpa}$.
We construct an adversary $\adv$ for experiment $\mathrm{Exp}_{\adv,\eseds_{\ppke}}^{\indcpads}$ as follows.

\begin{pchstack}[center]

\procedure
 {$\adv(\secparam, \pk)$}{
 \Angle{m_0, m_1, \state'} \gets \bdv(\secparam, \pk) \\
 \Mms_0 := \{ m_0 \} \\
 \Mms_1 := \{ m_1 \} \\
 \state := \state' \concat \{ \Mms_0, \Mms_1 \} \\
 \pcreturn \Angle{\Mms_0, \Mms_1, \state}
}

\pchspace

\procedure
 {$\adv(\secparam, \pk, \Cds, \state)$}{
 \state' \concat \{ \Mms_0, \Mms_1 \} := \state \\
 b' \gets \bdv(\secparam, \pk, \Cds[0], \state') \\
 \pcreturn \Angle{0, \Mms_{b'}[0]}
}

\end{pchstack}

The adversary $\bdv$'s view is indistinguishable from experiment $\mathrm{Exp}_{\bdv,\ppke}^{\indcpa}$
If adversary $\bdv$ guesses correctly, then $\adv$'s output is also correct.
Hence, if $\bdv$'s advantage is $\epsilon$, then $\adv$'s advantage is $\epsilon$.
\end{proof}

However, we also prove that even if each cell in $\Cds$ consists of a single ciphertext from a $\indcpa$-secure, public-key encryption scheme $\ppke$, then this does not imply $\indcpads$-security.
We prove by giving a data structure that consists of a single ciphertext from $\ppke$, but that is not $\indcpads$-secure.
Again, the proof for symmetric encryption is analogous.

\begin{theorem}
\label{thm:no-ind-cpads}
If each ciphertext in an efficiently searchable, encrypted data structure $\eseds_{\ppke}$ $\Cds[j]$ $(0 \leq j < n)$ is from a $\indcpa$-secure encryption scheme $\ppke$, then $\eseds_{\ppke}$ is not necessarily $\indcpads$-secure.
\end{theorem}

\begin{proof}
Given a multi-set of plaintexts $\Mms$ and a $\indcpa$-secure, public-key encryption scheme $\ppke$, we construct a data structure as follows.
Let $\operatorname{rand-order(m_i)}$ be the randomized order of each plaintext $m_i \in \Mms$.
Recall that in a randomized order of a multi-set, elements are sorted, but ties are broken based on the outcome of a coin flip.
$$
\Cds[\operatorname{rand-order}(m_i)] \gets \ppke.\enc(\pk, m_i)
$$

This data structure has equivalent leakage to frequency-hiding order-preserving encryption (FH-OPE) by Kerschbaum \cite{Ker15}.
It is easy to see that each cell of the data structure consists of only one semantically secure ciphertext.
However, we construct an adversary that succeeds with probability $1$ for $p = \dfrac{1}{2}$ in our experiment $\mathrm{Exp}_{\adv,\eseds_{\ppke}}^{\indcpads}$.
\begin{pchstack}[center]

\procedure
 {$\adv(\secparam, \pk)$}{
 \Mms := \{ 0, 1 \} \\
 \pcreturn \Angle{\Mms, \emptystring}
}

\pchspace

\procedure
 {$\adv(\secparam, \pk, \Cds, \state)$}{
 \pcreturn \Angle{0, 0}
}

\end{pchstack}

The adversary always wins the game, since in the given encryption scheme plaintext $0$ will always be encrypted at position $0$.
Grubbs et al.~showed in \cite{GruSek16} the practicality of the attack by constructing a plaintext guessing attack -- the bucketing attack described in Section~\ref{sec:bucket-attack} -- on FH-OPE.
In their experiments it succeeds with probability $30\%$ where the base line guessing probability is only $4\%$.
\end{proof}

\subsubsection{Relation to Other Security Definitions}

In searchable encryption a security definition of indistinguishability under chosen-keyword attack ($\indcka$-security) has been defined in \cite{CurGar11} and used in many subsequent works.
Loosely speaking, this security definition states that the data structure is $\indcka$-secure, if it is indistinguishable from a simulator given (a set of) leakage function(s) $\Leak$.
However, this can be misleading, since the leakage function does not necessarily clearly state the impact on plaintext guessing attacks.
We first state the following corollary:

\begin{corollary}
\label{cor:ind-cpa}
If a public-key encryption scheme $\ppke$ is \emph{\indcpa-secure}, then there exists a simulator $\simulator_{\ppke}(\secparam,\pk)$, such that for all PPT adversaries $\adv$ and all PPT distinguishers $\distinguisher$

\begin{multline*}
 \advantage{\indcpa}{\adv, \distinguisher, \ppke} := \\ 
   \left| \prob{\distinguisher(c, \pk) = 1 : \Angle{c, \pk} \gets \mathrm{RealExp}_{\adv,\ppke}^{\indcpa}(\secpar)} - \right. \\
   \left. \prob{\distinguisher(c, \pk) = 1 : \Angle{c, \pk} \gets \mathrm{SimExp}_{\simulator_{\ppke},\ppke}^{\indcpa}(\secpar)} \right| \ \ \\
  \leq \negl
\end{multline*}

\begin{pchstack}[center]
\procedure
 {$\mathrm{RealExp}_{\adv,\ppke}^{\indcpa}(\secpar)$}{
 \Angle{\pk, \sk} \gets \ppke.\kgen(\secparam) \\
 \Angle{m_0, m_1, \state} \gets \adv(\secparam, \pk) \\
 b \sample \bin \\
 c \gets \ppke.\enc(\pk, m_b) \\
 \pcreturn c, \pk
}

\pchspace

\procedure
 {$\mathrm{SimExp}_{\simulator_{\ppke},\ppke}^{\indcpa}(\secpar)$}{
 \Angle{\pk, \sk} \gets \ppke.\kgen(\secparam) \\
 c \gets \simulator_{\ppke}(\secparam,\pk) \\
 \pcreturn c, \pk
}

\end{pchstack}
\end{corollary}

It follows that there exists a simulator for an encrypted data structure whose cells consists only of semantically secure ciphertexts which requires a leakage function of only the length $n$ of the data structure and the public key $\pk$.
However, as we have shown in Theorem~\ref{thm:no-ind-cpads} such a data structure may not be $\indcpads$-secure and susceptible to plaintext guessing attacks.

\begin{theorem}
\label{thm:sim}
An efficiently searchable, encrypted data structure $\eseds_{\ppke}$ may be indistinguishably simulated with a leakage function $\Leak = \{ \pk, n \}$ and be susceptible to plaintext guessing attacks.
\end{theorem}

\begin{proof}
Consider the data structure from the proof of Theorem~\ref{thm:no-ind-cpads}.
It is indistinguishable from $n$ ciphertexts produced using public key $\pk$ and successful plaintext guessing attacks have been shown by Grubbs et al.~in \cite{GruSek16}.
\end{proof}

Hence, leaking the number of plaintexts may be sufficient for a successful plaintext guessing attack in a simulation-based security proof.
Our $\indcpads$-security model prevents this by introducing a {\em structural independence} constraint.
While Curtmola et al.~have been careful not to make this mistake in \cite{CurGar11} and their $\eseds$ is $\indcpads$-secure, subsequent work was not as careful.
Boelter et al.'s data structure \cite{BoePod16} has a (correct) simulation-based proof and is not $\indcpads$-secure and susceptible to plaintext guessing attacks.\footnote{This is easy to see, since they do not encrypt the structural information in their data structure, i.e.~the pointers to leaf nodes in the tree, and hence the ciphertexts can be ordered.}

\subsubsection{Impact on Plaintext Guessing Attacks}

We can now revisit the plaintext guessing attacks on deterministic and order-preserving encryption.
First, our security model fully captures the attack setup.
The adversary is given full ciphertext information and can chose the plaintexts such that it has perfect background knowledge\footnote{Recall that the adversary is allowed to submit the same plaintext multi-sets in the $\indcpads$-security experiment}, i.e.~the adversary in our model has at least the same information as was used in those attacks.
Second, our security definition implies that if the adversary is then able to infer even one plaintext better than with negligible probability over guessing our scheme is broken.
Hence security in the $\indcpads$ model implies security against all (passive, ciphertext-only) plaintext guessing attacks.

\section{An $\indcpads$-Secure $\eseds$ for Range Queries}
\label{sec:scheme}

We next present our efficiently searchable, encrypted data structure for range queries that is $\indcpads$-secure.
We emphasis that using the result from the data structure we can perform range queries in any commodity database management system without modifications.
Hence, our data structure is as easy to integrate as order-preserving encryption, yet it is secure against chosen-plaintext attacks.
We begin by describing the system architecture and give the intuition of our construction.
We then present our encryption algorithm and interactive search protocol.

\subsection{System Architecture}

\begin{figure}[!ht]
\centering
\includegraphics[width=\myimagesize]{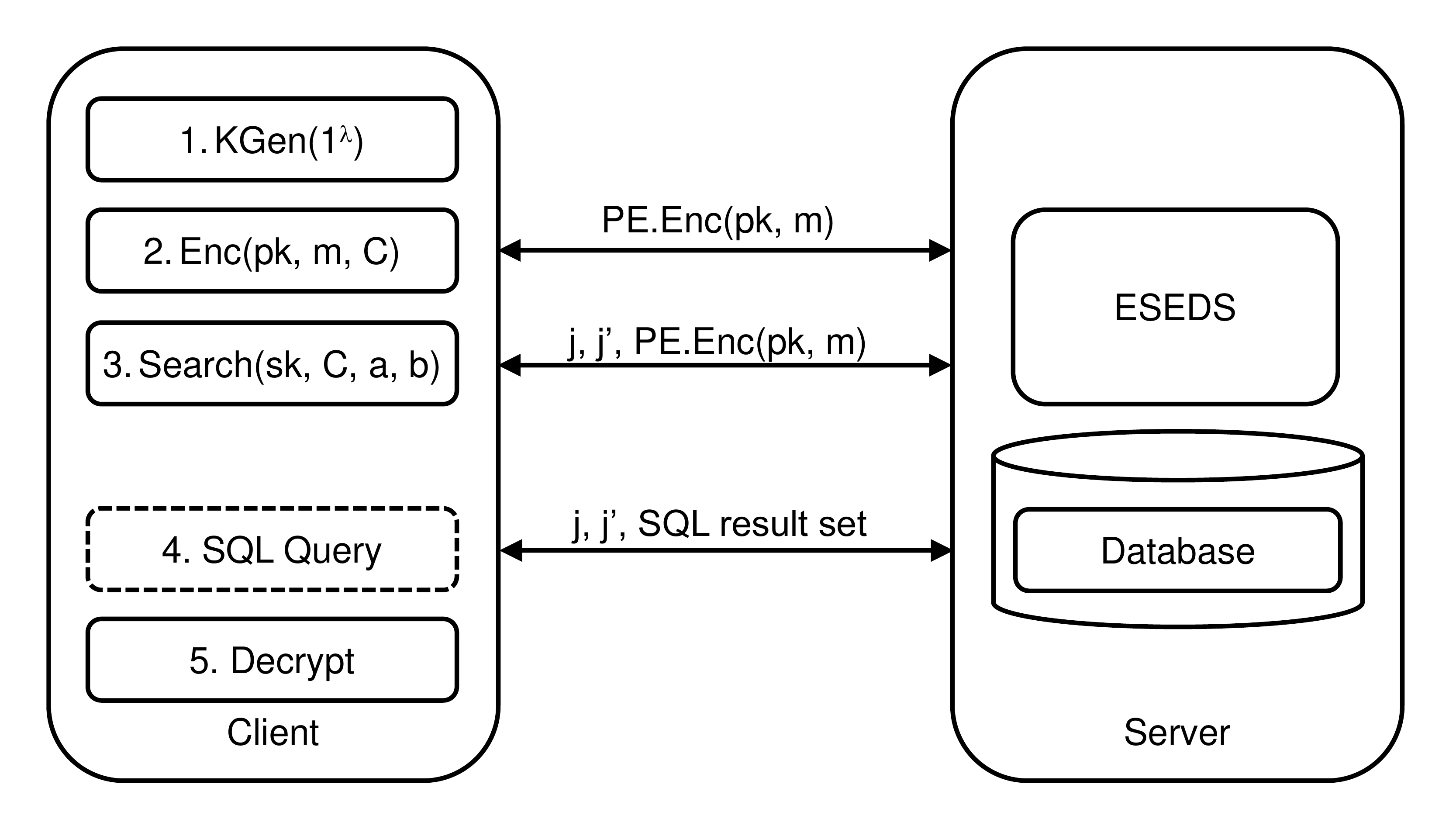}
\caption{Overview of system architecture}
\label{fig:overview}
\end{figure}

We depict an overview of our architecture in Figure~\ref{fig:overview}.
In our setup we assume a client holding the secret key $\key \gets \eseds$.\allowbreak{}$\kgen(\secparam)$ and a server that holds the data structure $\Cds$.
The server may hold several data structures managed independently for each database column, but needs to take care of correlation attacks as in \cite{DurDuB16}.
A database table then contains the rows linking the entries by their index in the data structure.
After encrypting the plaintexts, the client and the server can interactively perform a search query, e.g.~a range query, on the server's data structure which results in two indices $j, j'$.
%
Then these two indices $j, j'$ can be used in subsequent range queries on the database management system.
We assume that the server is {\em semi-honest}, i.e.~only performs {\em passive} attacks.
This model is commonly assumed in the scientific literature on database security.

\subsection{Intuition}
\label{sec:intuition}

Our data structure combines the ideas of three previous order-preserving encryption schemes:
First, the scheme by Popa et al.~\cite{PopLi13} provides the basis for managing the order of ciphertexts in a stateful, interactive manner.
Of course, this scheme is not secure against the attacks by Naveed et al., since it is deterministic and ordered.
Second, we add the frequency-hiding aspect of the scheme by Kerschbaum~\cite{Ker15}.
The scheme itself cannot be used as the basis of an $\indcpads$-secure data structure, since it partially leaks the insertion order.
Therefore the frequency-hiding idea needs to be fit into Popa et al.'s scheme.
We do this by encrypting the plaintext using a probabilistic algorithm (similar to the stOPE scheme in~\cite{PopLi13}) and also inserting a ciphertext for each plaintext using Kerschbaum's random tree traversal.
This combined construction would still not be $\indcpads$ secure.
Third, we apply Boldyreva et al.'s modular order-preserving encryption idea~\cite{BolChe11}.
This idea rotates the plaintexts around a modulus statically hiding the order.
However, modular order-preserving encryption has been developed for {\em deterministic} order-preserving encryption.
In our probabilistic encryption -- as introduced by Kerschbaum -- we need to apply the modulus on the ciphertexts.
This can be done by updating the modulus after encryption.

In summary, intuitively our encryption scheme works as follows:
We maintain a list of ciphertexts for each plaintext (including duplicates) sorted by the plaintexts on the server.
However, the list is rotated around a random offset (chosen uniformly from the range between $1$ and the number of ciphertexts).
We then encrypt and search using binary search.
However, due to the rotation which can divide a set of identical plaintexts adjacent in the list into a lower and upper part, the search and encryption algorithms become significantly more complex which is apparent in their detailed description below.

\subsection{Encryption Algorithm}
\label{sec:enc}

Let $\pe$ be a standard, probabilistic encryption scheme supporting the following three -- possibly probabilistic -- polynomial-time algorithms: $\kgen$, $\enc$ and $\dec$.
We use symmetric encryption, e.g.~AES in CBC or GCM mode, for speed, but assume an encryption oracle in the definition of semantic security.
Let $\mathbb{D}$ be the domain of plaintexts and $N = |\mathbb{D}|$ its size.
We now describe the algorithms and protocols of our efficiently searchable, encrypted data structure:

\begin{itemize}

\item $\key \gets \kgen(\secparam)$: Execute $\key \gets \pse.\kgen(\secparam)$.

\item $\Cds^{h+1} \gets \enc(\key, m, \Cds^h)$:
We denote $\Cds^h$ as $\Cds$ for brevity, if it is clear from the context.
First the client and server identify the index $j_m$ where $m$ is to be inserted (before).
Then the client sends the ciphertext of $m$ to the server which inserts it at position $j_m$.
Finally the server rotates the data structure by a random offset.

\begin{enumerate}

\item The client sets $l := 0$ and $u := n - 1$.

\item If $n = 0$ then go to step \ref{stp:sendc}.

\item The client requests $\Cds[0]$ and sets $r := \dec(\key, \Cds[0])$.

\item 
Set $j := \lfloor l + \frac{u - l}{2} \rfloor$.
The client requests $\Cds[j]$ and executes $m' := \dec(\key, \Cds[j])$.
If $m' - r \bmod N > m - r \bmod N$, then the client sets $l := j + 1$.
If $m' - r \bmod N < m - r \bmod N$, then the client sets $u := j$.
If $m' = m \bmod N$, then the client flips a random coin and sets either $l := j + 1$ or $u := j$ depending on the outcome of the coin flip.
The client repeats this step until $l = u$.

\item 
\label{stp:sendc}
The client sends $c \gets \pse.\enc(\key, m)$ to the server.

\item The server sets
\begin{eqnarray*}
 \Cds[n] & := & \Cds[n-1] \\
 \Cds[n-1] & := & \Cds[n-2] \\
 \cdots\\
 \Cds[l+1] & := & \Cds[l] \\ 
 \Cds[l] & := & c
\end{eqnarray*}
The server sets $n := n+1$.

\item The server chooses a random number $s \sample \mathbb{Z}_{n-1}$.
The server sets the new encrypted data structure to $\Cds^{h+1}[j] := \Cds[j+s \bmod n]$ for $0 \leq j < n$ as a result of the encryption operation.
This data structure $\Cds^{h+1}$ will be used as input to the next encryption operation.

\end{enumerate}

\item $\Angle{j, j'} := \search(\key, \Cds, a, b)$:
Wlog.~we assume that $a \leq b$ in the further exposition.
In case $a > b$ the query is rewritten as to match all $x$, such that $0 \leq x < b \lor a \leq x < N$.

Let $j_{\mathrm{min}}(v)$ be the minimal index of plaintext $v$ and $j_{\mathrm{max}}(v)$ be the maximal index of plaintext $v$.
$$j_{\mathrm{min}}(v) := \min(j | \dec(\key, \Cds[j]) = v)$$
$$j_{\mathrm{max}}(v) := \max(j | \dec(\key, \Cds[j]) = v)$$

If $j_{\mathrm{min}}(v) = 0$ and $j_{\mathrm{max}}(v) = n - 1$ and there are two distinct plaintexts in the data structure, then we redefine as 
$$j_{\mathrm{min}}(v) := j + 1 | \dec(\key, \Cds[j]) < v \land \dec(\key, \Cds[j+1]) = v)$$
$$j_{\mathrm{max}}(v) := j - 1 | \dec(\key, \Cds[j]) > v \land \dec(\key, \Cds[j-1]) = v)$$

If $a$ and $b$ do not span the modulus, i.e.~$j_{\mathrm{min}}(a) < j_{\mathrm{max}}(b)$, then a query for $x \in [a, b]$ is rewritten to $j_{\mathrm{min}}(a) \leq x \leq j_{\mathrm{max}}(b)$.
Else, it is rewritten to $0 \leq x < j_{\mathrm{max}}(b) \lor j_{\mathrm{min}}(a) \leq x < n'$.

Both $j_{\mathrm{min}}(a)$ and $j_{\mathrm{max}}(b)$ are found using a separately run, interactive binary search.
We next present this protocol.

\begin{enumerate}

\item The client sets $l := 0$ and $u := n - 1$.

\item The client requests $\Cds[0]$, $\Cds[n-1]$ and sets $r := \dec(\key, \Cds[0])$.
If $\dec(\key, \Cds[0]) = \dec(\key, \Cds[n-1])$ and searching for $j_{\mathrm{min}}(a)$, it sets $r := r + 1$.

\item Set $j := \lfloor l + \frac{u - l}{2} \rfloor$.
The client requests $\Cds[j]$ and executes $m := \dec(\key, \Cds[j])$.
If $m - r \bmod N < a - r \bmod N$ (or $m - r \bmod N \leq b - r \bmod N$, respectively) then the client sets $l := j + 1$.
Else the client sets $u := j$.
The client repeats this step until $l = u$.

\item The client returns $j_{\mathrm{min}}(a) := l$ (or $j_{\mathrm{max}}(b) := u$, respectively).

\end{enumerate}

\item $m := \dec(\key, \Cds[j])$:
Set $m := \pse.\dec(\key, \Cds[j])$.

\end{itemize}

\subsection{Security}

\begin{theorem}
\label{thm:slf-ope}
Our efficiently searchable, encrypted data structure $\eseds_{\pse}$ is $\indcpads$-secure.
\end{theorem}

\begin{proof}

Since all cells of the data structure consists only of ciphertexts from a $\indcpa$-secure encryption scheme, we can replace the encrypted data structure by a simulator.
Let the simulator $\simulator_{\eseds}^{\mathrm{E}_{\key}}(\secparam, n)$ output $n$ ciphertexts $c \gets \mathrm{E}_{\key}(0)$.
The adversary $\adv$ cannot distinguish the following experiment $\mathrm{Exp}_{\adv,\simulator_{\eseds},\eseds_{\pse}}^{\indcpads}$ from experiment $\mathrm{Exp}_{\adv,\eseds_{\pse}}^{\indcpads}$ except with negligible probability.

\begin{pcvstack}[center]
\procedure
 {$\mathrm{E}_{\key}(m)$}{
  c \gets \pse.\enc(\key, m) \\
  \pcreturn c
}

\pcvspace

\procedure
 {$\mathrm{Exp}_{\adv,\simulator_{\eseds},\eseds_{\pse}}^{\indcpads}(\secpar)$}{
 \Angle{\key, \Cds} \gets \eseds.\kgen(\secparam) \\
 \Angle{\Mms_0, \Mms_1, \state} \gets \adv^{\mathrm{E}_{\key}}(\secparam) \\
 \pcif |\Mms_0| \neq |\Mms_1| \pcthen \ \pcreturn \bot \\
 \Cds \gets \simulator_{\eseds}^{\mathrm{E}_{\key}}(\secparam, |\Mms_0|) \\
 \Angle{j', m'} \gets \adv^{\mathrm{E}_{\key}}(\secparam, \Cds, \state) \\
 \pcreturn \Angle{\eseds.\dec(\Cds[j']) = m', \dfrac{\#_{\Mms_0 \cup \Mms_1}m'}{|\Mms_0 \cup \Mms_1|}}
}
\end{pcvstack}

The adversary $\adv$ in $\mathrm{Exp}_{\adv,\simulator_{\eseds},\eseds_{\pse}}^{\indcpads}$ clearly has no information which plaintext multi-set has been encrypted or about the plaintexts' positions in the data structure.
Since in our $\eseds_{\pse}$ each plaintext has equal probability of being at any index within the data structure, the adversary can at best guess the index $j'$ for any $m' \in \Mms$.
However, the probability of a successful guess is bounded by $\dfrac{\#_{\Mms_0 \cup \Mms_1}m'}{|\Mms_0 \cup \Mms_1|}$.
\end{proof}

\subsection{Implementation}

In order to allow efficient online encryptions, we employ a technique we call {\em decoupled encryptions} in our implementation that however temporarily violates $\indcpads$-security.
A decoupled encryption has the positive effect that an encryption operation returns control almost instantly to the client.
First, we store the index $j$ explicitly in a database table along with the ciphertexts.
Second, we choose a large domain $D$ for the index, e.g.~256 bit.
When we insert a new plaintext $m$ into the data structure, we search for the element $\Cds[j']$ before and the element $\Cds[j'']$ after the new element as described before.
Then we insert $m$ as $\Cds[\lfloor \dfrac{j''-j'}{2} \rfloor + j'] = \pse.\enc(k, m)$.
This operation is constant time, however after multiple encryptions the adversary may distinguish the data structures for two distinct sets of plaintexts.

To restore security, we operate a background process in the database management system.
This background process scans the entire data structure and makes the indices of all neighbouring data cells equidistant.
For example, let there be $n$ ciphertexts in the data structure and let $|D|$ be the size of the domain of the index.
Then the background process assigns the indices $\lfloor \dfrac{|D|}{n+1} \rfloor$, $2 \cdot \lfloor \dfrac{|D|}{n+1} \rfloor$, $3 \cdot \lfloor \dfrac{|D|}{n+1} \rfloor$, $\ldots$ to the data cells.
The background process also rotates the data structure around a new random number $r$.
After the background process completes the data structure is $\indcpads$-secure.

The background process can run incrementally and independently of queries.
This allows it to be scheduled adaptively to the load of the database system.
Hence, decoupled encryptions allow efficient search and online encryption operations while reaching $\indcpads$-security eventually.

\section{Performance Evaluation}
\label{sec:evaluation}

We prototypically implemented and in a number of experiments evaluated the performance our $\indcpads$-secure $\eseds$.
In this section we report the results of our experiments measuring the run-time of range searching over encrypted data.

\subsection{Implementation}

We used Java for our implementation and evaluation, since many multi-tier applications are implemented in Java.
Although a native cryptographic library, such as Intel's AES-NI, promises further performance improvements, programming languages such as C or C++ are more commonly used for systems software (such as database management systems) rather than for database applications (which only issue database queries).
However, in our setup encryption and decryption is performed in the database application.
We used Oracle's Java 1.8 and all experiments were run on the Java SE 64-Bit Server virtual machine.
The database backend was the MySQL replacement MariaDB in version 10.1.
When using a database, such as MariaDB, that was not specifically developed for operation on encrypted data, one needs to configure it to prevent the attacks on configuration described by Grubbs et al.~\cite{GruRis17}.
All experiments were run on a single machine with a 4-core Intel i7 CPU at 2.9 GHz and 16 GB of RAM on Windows 10 Enterprise.

\subsection{Experimental Setup}

We measure the run-time of a typical, simply structured (i.e.~a single search term and no conjunctions or disjunctions) database query on a single ordered database column, e.g.~a range query or a top-k query.
We use synthetic data and queries.
However, we adapt our choice of parameters to the data from the DBLP data set.
In the spirit of Grubbs et al.~\cite{GruSek16} we considered author names.
At the time of our experiments there were about $1.500.000$ million distinct author names in DBLP, the most frequent of which appears roughly $80$ times.

We implement the client interface as it would be used in an application using a database.
The application supplies the parameters, e.g.~the start $a$ and end $b$ of a range or the $k$ in top-k, and receives the results in plaintext.
Thus, our measured run-time includes the $\search$ algorithm, the standard query by the database management system and the decryption of the result.
We emphasize that in more complex queries, e.g.~including multiple search terms combined by conjunction and disjunctions, the relative time for executing the query on the database management system would be proportionally higher.
Hence, our experiments put an upper bound on the worst case of the proportional overhead.

Our target quantity in our measurements is the absolute run-time in milliseconds.
For range queries we measure the dependence of the run-time on different parameters.
\begin{itemize}

\item {\em Size of the database}:  We vary the database size from $100.000$ to $1.000.000$ plaintexts in steps of $100.000$, i.e.~data items before encryption.

\item {\em Size of the queried range}: We vary the range size and consequently the result set size in the query from $10$ to $100$ in steps of $10$.

\end{itemize}

For top-$k$ queries we measure the dependence of the run-time of the following parameter.
\begin{itemize}

\item {\em $k$}:  We vary the limit $k$ from $10$ to $100$ in steps of $10$.

\end{itemize}

We compare the run-time on encrypted data to the run-time on plaintext data.
Note that queries on plaintext only need to execute the query on the database management system, i.e.~the time for the $\search$ algorithm and decryption of results is $0$.

We use synthetically generated data and queries.
We uniformly choose distinct plaintexts and we uniformly choose a begin of the range query and then compute the end using the fixed size parameter of the experiment.

We repeat each experiment $30$ times discarding the first $10$ experiments in order to allow to adjust the Java JIT compiler.
We report the mean and $95\%$ confidence interval for each parameter setting.

\subsection{Results}

\begin{figure}[!h]
\centering
\includegraphics[width=\myimagesize]{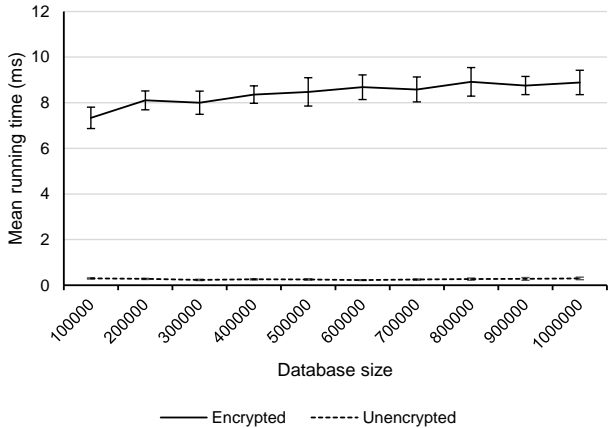}
\caption{Performance over different database sizes}
\label{fig:size}
\end{figure}

{\em Database size}: Figure~\ref{fig:size} shows the running time over the database size.
We use a query range size of $10$.
The database size increases from $100.000$ to $1.000.000$ plaintexts in steps of $100.000$.
The running time is measured in milliseconds.
The error bars show the $95\%$ confidence interval.
Since our search algorithms run in sub-linear time only a very slight increase ($20\%$) in running time is measurable compared to the increase in database size ($900\%$).
The overhead of our encryption is roughly $9$ milliseconds.

\begin{figure}[!h]
\centering
\includegraphics[width=\myimagesize]{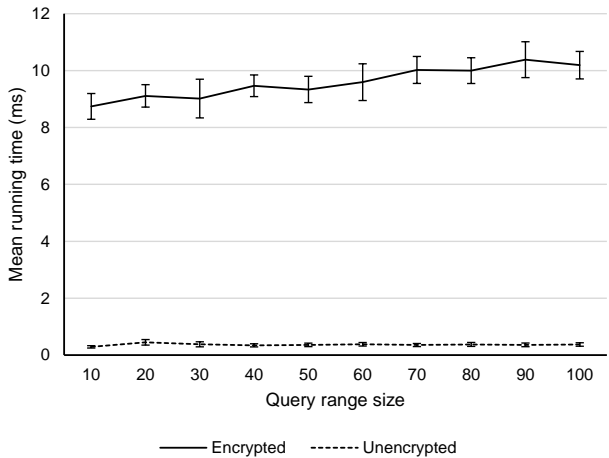}
\caption{Performance over query range size}
\label{fig:range}
\end{figure}

{\em Query Range Size}: Figure~\ref{fig:range} shows the running time over the query range size.
We use a database size of $1.000.000$ plaintexts.
The query range size and hence the expected result set size increases from $10$ to $100$ in steps of $10$.
The running time is measured in milliseconds.
The error bars show the $95\%$ confidence interval.
The running time increase is slight and approximately linear in the query range size and there is a constant baseline.
We attribute the constant cost to our binary search algorithm which as shown in Figure \ref{fig:size} behaves almost constant for these database sizes.
We attribute this increase to the cost of decryption which is dominated by the cryptographic operations.

\begin{figure}[!h]
\centering
\includegraphics[width=\myimagesize]{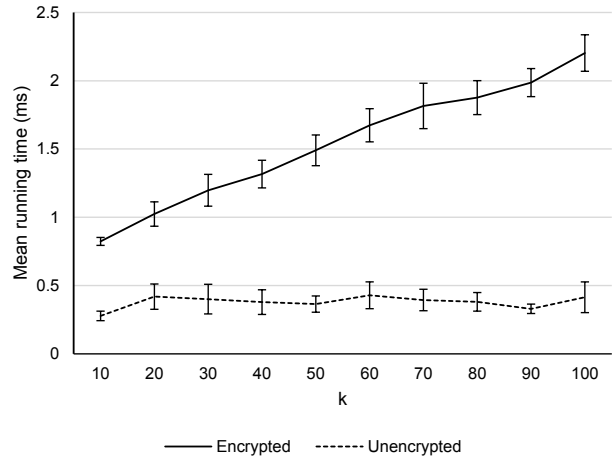}
\caption{Performance over k for top-k queries}
\label{fig:topk}
\end{figure}

{\em Top-k queries}: Figure~\ref{fig:topk} shows the running time for top-$k$ queries over $k$.
We use a database size of $1.000.000$ plaintexts.
The value of $k$ increases from $10$ to $100$ in steps of $10$.
The running time is measured in milliseconds.
The error bars show the $95\%$ confidence interval.
The constant baseline is lower, since top-$k$ queries can be executed without a search algorithm, when the minimum ciphertext (the rotation value) is stored as part of the key.
The linear increase due to decryption of results is now clearly visible.

\subsection{Discussion}

We observe an almost constant overhead for the search algorithm of about $9$ milliseconds.
Then, decryption and filtering is linear in the result size.
However, for reasonable result set sizes -- up to $100$ ciphertexts in our experiments -- it stays below $2$ milliseconds.
Note that decryption is unavoidable in searching over encrypted data and often excluded in other scientific work.
The database query time is not measurably affected by our scheme.

Hence, we conclude that our $\indcpads$-secure data structure adds an overhead of roughly $10$ milliseconds (per encrypted column in the query) for reasonable result sizes.
This is a very good performance compared to other schemes for range queries over encrypted data \cite{BoePod16,DemPap16,HahKer16,LewWu16,RocApo16}.

\section{Related Work}
\label{sec:related}

Our work is related to other order-preserving encryption schemes, searchable encryption schemes -- particularly for range queries --, order-revealing encryption, leakage-abuse attacks and other encryption schemes that in principle can be used to perform range queries.

\subsection{Order-Preserving Encryption}

Order-preserving encryption was introduced by Agrawal et al.~in \cite{AgrKie04}.
The idea is based on running queries using unmodified database management systems using deterministic encryption by Hacig\"um\"us et al.~\cite{HacIye02a}.
However, Agrawal et al.~extended it to range queries.
Their original proposal uses an informal security model.
Later, Boldyreva et al.~\cite{BolChe09} formalized the security and presented a new construction.
They define indistinguishability under ordered chosen plaintext attack ($\indocpa$).
Note that $\indocpa$-security leaks the order of plaintexts by the ciphertexts and is hence strictly less secure than $\indcpads$-security.
They also show that with constant local storage (the key only) $\indocpa$-security requires exponentially sized ciphertexts and therefore settle for a weaker notion.
Again, Boldyreva et al.~further improve this definition in \cite{BolChe11}.
In this paper, they also introduce modular order-preserving encryption.
Mavroforakis et al.~show how to improve the security of modular order-preserving encryption against query observation attacks in \cite{MavChe15}.
They introduce fake queries to hide the modulus, but this only works for uniformly distributed plaintexts as shown by Durak et al.~\cite{DurDuB16}.
An improved security model and construction requiring only constant local storage was introduced by Teranishi et al.~in \cite{TerYun14}.
Their idea is to occasionally introduce larger gaps into the ciphertexts.
However, this also does not yet achieve $\indocpa$ security and has not been tested against the attacks by Naveed et al.
Hwang et al.~present a performance improvement for this encryption scheme in \cite{HwaKim15}.
They show how to use a more efficient random sampling.

The first $\indocpa$ secure order-preserving encryption was presented by Popa et al.~in \cite{PopLi13}.
It also forms the first $\eseds$ for order-preserving encryption, since it imposes a data structure beyond a single ciphertext.
They introduce the concept of storing the state (symmetrically encrypted) on the server and make ciphertexts (necessarily) mutable, i.e.~adapting to insertions.
Schr\"opfer and Kerschbaum improve the performance of this model in~\cite{KerSch14}.
Kerschbaum introduces an even stronger security model -- indistinguishability under frequency-analyzing ordered chosen plaintext attack -- in \cite{Ker15}.
We build upon his idea and incorporate the concept of frequency-hiding into our data structure.
Roche et al.~combine FH-OPE by Kerschbaum with on-demand sorting, i.e.~when searches are performed \cite{RocApo16}.
While their encryption is strongly secure before any queries, it deteriorates after queries even on the stored data structure and hence is less secure than $\indcpads$.


There are many more order-preserving encryption schemes that have been proposed in the literature \cite{AgrElA09,HilKos11,KadAma10a,KadAma10b,KreYak15,LeePar09,LiZha15,LiuWan12,LiuWan13,OezSin03,WozRos13,XiaYen12a,XiaYen12b} which we do not discuss here, since they lack a formal security analysis.

\subsection{Searchable Encryption}

Searchable encryption allows the comparison of a token (corresponding to a plaintext) to a ciphertext.
The ciphertext (without any token) is $\indcpa$-secure.
The token can match plaintexts for equality or the plaintext to a range.
Only the secret key holder can create tokens.

The concept of searchable encryption has been introduced by Song et al.~in \cite{SonWag00}.
It supports equality searches and additions, but requires linear time for searching, since each ciphertext needs to be compared.
In order to speed up search an encrypted inverted index can be built.
This inverted index is an $\eseds$, since it imposes a data structure.
The first encrypted inverted index for equality search was presented by Curtmola et al.~in \cite{CurGar11}.
It is an efficiently searchable, encrypted data structure.
It supports (expected) constant time search, but all plaintexts (the inverted index) need to be encrypted at once and additions are not supported.
Dynamic searchable encryption \cite{KamPap12} made the data structure mutable in order to support additions.
Since then a number of dynamic searchable encryption schemes with indexes have been proposed \cite{CasJae14,CasJar13,HahKer14,KamPap13,NavPra14,StePap14}.
A recent survey provides a good overview \cite{BosHar14}.

Tackling range queries with searchable encryption is more complex.
The first proposal by Boneh and Waters in \cite{BonWat07} had ciphertext size linear in the size of the domain of the plaintext.
The first poly-logarithmic sized ciphertexts scheme was proposed by Shi et al.~in \cite{ShiBet07}.
However, their security model is somewhat weaker than standard searchable encryption.
The construction is based on inner-product predicate encryption which has been made fully secure by Katz et al.~in \cite{KatSah08}.
All schemes follow the construction by Song et al.~(without inverted indices) and require linear search time.
The first attempt to build range-searchable encryption into an index (an $\eseds$) has been made by Lu in \cite{Lu12}.
However, the inverted index tree reveals the pointers and is hence no more secure than order-preserving encryption.
Demertzis et al.~\cite{DemPap16} map a range query to keyword queries by providing tradeoffs between storing replicated values in each of its ranges and enumerating all values within range query.
The search can then be easily performed using the data structure of Curtmola et al.~\cite{CurGar11}.
While the scheme is range searchable, its queries are very revealing and it has high storage cost (at least $O(n \log n)$).
Boelter et al.~\cite{BoePod16} use garbled circuits to implement the search within a node of the index.
They do not encrypt the pointers in the index and are hence susceptible to the attacks by Naveed et al.~and are not $\indcpads$ secure.
The scheme by Hahn and Kerschbaum~\cite{HahKer16} creates an index using the access pattern of the range queries.
No other information is leaked, however, this provides amortized poly-logarithmic search time.
The scheme is only $\indcpads$-secure as long as no queries have been performed (and the index has been partially built).
Their scheme is based on inner-product predicate encryption which is too slow for practical use.

\subsection{Order-Revealing Encryption}

Order-revealing encryption~\cite{BonLew15,CheLew16,LewWu16} is an alternative to order-preserving encryption.
Instead of preserving the order there is a public function that reveals the order of two plaintexts using the ciphertexts only.
At first, it may seem paradoxical to combine the disadvantages of order-preserving and searchable encryption: order revelation and modified comparison function.
However, order-revealing encryption has also advantages.
It allows an $\indocpa$ secure encryption with constant-size ciphertexts, constant size client storage and without mutation circumventing impossibility results in \cite{BolChe09} and \cite{PopLi13}.
However, the first construction was not only impractical due to its disadvantages, but also due to its performance.
A different construction with slightly more leakage, but significantly better performance was presented by Chenette et al.~in \cite{CheLew16}.
This construction was further improved by Lewi and Wu in \cite{LewWu16}.
They allow comparison only between a token and an $\indcpa$-secure ciphertext as in searchable encryption, i.e.~the scheme has no leakage when no token is revealed.
Their search procedure requires a linear search over all ciphertext and no indexing is possible.
Hence, compared to our scheme which has logarithmic search time, order-revealing encryption currently remains impractical.

\subsection{Leakage-Abuse Attacks}

We discussed many leakage-abuse attacks on search over encrypted data.
There are static attacks on order-preserving encryption \cite{DurDuB16,GruSek16,NavKam15,PouWri16} and attacks using dynamic information that also work on searchable encryption \cite{CasGru15,GruMcP16,IslKuz12,KelKol16,LacMin17,ZhaKat16}.

Kellaris et al.~\cite{KelKol16} have presented generic inference attacks on encrypted data using range queries.
Their attacks work in a setup where the adversary has compromised the database server and can observe all queries, i.e.~they work for dynamic leakage during the execution of queries and are not ciphertext-only attacks.
They do not assume a specific cryptographic protection mechanism, but work only on its dynamic leakage profile, such as the access pattern or the result size, i.e.~they also apply to ORAM-protected databases.
The prerequisite assumption for Kellaris et al.'s attack to work is that the distribution of queries and the distribution of plaintexts differ.
Specifically, they assume that each possible query will be executed, but not each possible plaintext is in the database.
We note that Kellaris et al.~performed all their attacks on synthetic data and queries whereas static ciphertext-only attacks on real data have been publicized \cite{Duc13}.

There are also some more specific inference attacks.
Islam et al.~\cite{IslKuz12} and Cash et al.~\cite{CasGru15} have performed inference attacks by observing the queries on encrypted data.
Islam et al.~assume that the distribution of query keywords is approximately known and then can recover the query keywords using frequency analysis.
Cash et al.~improve the accuracy of this attack even under slightly weaker assumptions about the knowledge of query distribution, but then also use the information to recover plaintexts from the access pattern.
Lacharite et al.~\cite{LacMin17} improve the accuracy of plaintext guessing by incorporating information from observed queries.

Next to the ones already discussed plaintext guessing attacks Pouliot and Wright show that adding deterministic encryption to Bloom filters -- not surprisingly -- does not prevent cryptanalysis \cite{PouWri16}.
Zhang et al.~assume that the adversary can actively insert plaintexts and can then recover query plaintexts from the access pattern \cite{ZhaKat16}.
Grubbs et al.~\cite{GruMcP16} also attacked an implementation of multi-user searchable encryption which allows inferences between users leading to a complete breakdown of the security guarantee of encrypted web applications.

\subsection{Other Encryption Schemes}

Several attempts were made to build indexes for range queries using deterministic encryption or distance-revealing encryption \cite{HorMeh04,WanRav13}.
However, since they do not follow a formal security model and are based on primitives that are easily attackable we do not consider them here.

Oblivious RAM \cite{GolOst96} allows to hide the accesses to disk or memory and hence the access pattern of searchable or order-preserving encryption.
However, as Naveed showed in \cite{Nav15} the combination is not straightforward.
Recently, a new ORAM technique -- TWORAM -- has been presented by Garg et al.~in \cite{GarMoh16} that overcomes these limitations.
Kellaris et al.~showed in \cite{KelKol16} that inference attacks even against ORAM-protected range queries exist.

In theory search can be implemented without leakage using homomorphic encryption \cite{Gen09}.
However, since in our model the server returns an arbitrarily sized subset of the data and in homomorphic encryption the worst case determines the cost, the server would always return the entire encrypted data.
In terms of performance this can, of course, always be beaten by symmetric encryption and search on the client.

\section{Conclusions}
\label{sec:conclusions}

We present the $\indcpads$-security model -- an extension of semantic security -- that provably prevents plaintext guessing attacks as those by Naveed et al.~\cite{NavKam15} and Grubbs et al.~\cite{GruSek16}\intersentencespace
We show how this model implies that each ciphertext of an efficiently searchable, encrypted data structure must be semantically secure.
However, we also show that even if all ciphertexts in a data structure are semantically secure, this does not imply $\indcpads$-security.

Then we present an efficiently searchable (logarithmic time, linear space), encrypted data structure secure in this model.
We show that this scheme is practical in our evaluation, since it only has a $10$ milliseconds overhead on a range query over a million database entries.
This shows that one can built efficient, encrypted databases that withstand break-ins and data theft as we have seen in many recent attacks on cloud infrastructures.

\subsection{Future Work: Full dynamicity}

For ease of exposition we excluded deletion from the operations of our efficiently searchable, encrypted data structures $\eseds$.
However, given our instantiation for range queries over encrypted data, it should be easy to see that deletion does not pose any major obstacle compared to insertion.
Of course, for a fully functional database implementation we also implement deletion.

{
\bibliographystyle{IEEEtranS}
\bibliography{paper}

\begin{thebibliography}{10}
\providecommand{\url}[1]{#1}
\csname url@samestyle\endcsname
\providecommand{\newblock}{\relax}
\providecommand{\bibinfo}[2]{#2}
\providecommand{\BIBentrySTDinterwordspacing}{\spaceskip=0pt\relax}
\providecommand{\BIBentryALTinterwordstretchfactor}{4}
\providecommand{\BIBentryALTinterwordspacing}{\spaceskip=\fontdimen2\font plus
\BIBentryALTinterwordstretchfactor\fontdimen3\font minus
  \fontdimen4\font\relax}
\providecommand{\BIBforeignlanguage}[2]{{%
\expandafter\ifx\csname l@#1\endcsname\relax
\typeout{** WARNING: IEEEtranS.bst: No hyphenation pattern has been}%
\typeout{** loaded for the language `#1'. Using the pattern for}%
\typeout{** the default language instead.}%
\else
\language=\csname l@#1\endcsname
\fi
#2}}
\providecommand{\BIBdecl}{\relax}
\BIBdecl

\bibitem{AgrElA09}
D.~Agrawal, A.~{El Abbadi}, F.~Emek\c{c}i, and A.~Metwally, ``Database
  management as a service: challenges and opportunities,'' in \emph{Proceedings
  of the 25th International Conference on Data Engineering}, ser. ICDE, 2009.

\bibitem{AgrKie04}
R.~Agrawal, J.~Kiernan, R.~Srikant, and Y.~Xu, ``Order preserving encryption
  for numeric data,'' in \emph{Proceedings of the ACM International Conference
  on Management of Data}, ser. SIGMOD, 2004.

\bibitem{BoePod16}
T.~Boelter, R.~Poddar, and R.~A. Popa, ``A secure one-roundtrip index for range
  queries,'' IACR Cryptology ePrint Archive, Tech. Rep. 568, 2016.

\bibitem{BolChe09}
A.~Boldyreva, N.~Chenette, Y.~Lee, and A.~O'Neill, ``Order-preserving symmetric
  encryption,'' in \emph{Proceedings of the 28th International Conference on
  Advances in Cryptology}, ser. EUROCRYPT, 2009.

\bibitem{BolChe11}
A.~Boldyreva, N.~Chenette, and A.~O'Neill, ``Order-preserving encryption
  revisited: improved security analysis and alternative solutions,'' in
  \emph{Proceedings of the 31st International Conference on Advances in
  Cryptology}, ser. CRYPTO, 2011.

\bibitem{BonLew15}
D.~Boneh, K.~Lewi, M.~Raykova, A.~Sahai, M.~Zhandry, and J.~Zimmerman,
  ``Semantically secure order-revealing encryption: multi-input functional
  encryption without obfuscation,'' in \emph{Proceedings of the 34th
  International Conference on Advances in Cryptology}, ser. EUROCRYPT, 2015.

\bibitem{BonWat07}
D.~Boneh and B.~Waters, ``Conjunctive, subset, and range queries on encrypted
  data,'' in \emph{Proceedings of the 4th Theory of Cryptography Conference},
  ser. TCC, 2007.

\bibitem{BosHar14}
C.~B\"{o}sch, P.~Hartel, W.~Jonker, and A.~Peter, ``A survey of provably secure
  searchable encryption,'' \emph{ACM Computing Surveys}, vol.~47, no.~2, 2014.

\bibitem{CasGru15}
D.~Cash, P.~Grubbs, J.~Perry, and T.~Ristenpart, ``Leakage-abuse attacks
  against searchable encryption,'' in \emph{Proceedings of the 22nd ACM
  Conference on Computer and Communications Security}, ser. CCS, 2015.

\bibitem{CasJae14}
D.~Cash, J.~Jaeger, S.~Jarecki, C.~Jutla, H.~Krawczyk, M.~Rosu, and M.~Steiner,
  ``Dynamic searchable encryption in very-large databases: Data structures and
  implementation,'' in \emph{Proceedings of the 21st Network and Distributed
  System Security Symposium}, ser. NDSS, 2014.

\bibitem{CasJar13}
D.~Cash, S.~Jarecki, C.~Jutla, H.~Krawczyk, M.-C. Rosu, and M.~Steiner,
  ``Highly-scalable searchable symmetric encryption with support for boolean
  queries,'' in \emph{Proceedings of the 33rd Cryptology Conference}, ser.
  CRYPTO, 2013.

\bibitem{CheLew16}
N.~Chenette, K.~Lewi, S.~Weis, and D.~Wu, ``Practical order-revealing
  encryption with limited leakage,'' in \emph{Proceedings of the 23rd
  International Workshop on Fast Software Encryption}, ser. FSE, 2016.

\bibitem{CurGar11}
R.~Curtmola, J.~Garay, S.~Kamara, and R.~Ostrovsky, ``Searchable symmetric
  encryption: improved definitions and efficient constructions,'' \emph{Journal
  of Computer Security}, vol.~19, no.~5, 2011.

\bibitem{DemPap16}
I.~Demertzis, S.~Papadopoulos, O.~Papapetrou, A.~Deligiannakis, and
  M.~Garofalakis, ``Practical private range search revisited,'' in
  \emph{Proceedings of the ACM International Conference on Management of Data},
  ser. SIGMOD, 2016.

\bibitem{Duc13}
P.~Ducklin, ``Anatomy of a password disaster -- adobe's giant-sized
  cryptographic blunder,''
  \url{https://nakedsecurity.sophos.com/2013/11/04/anatomy-of-a-password-disaster-adobes-giant-sized-cryptographic-blunder/},
  2013.

\bibitem{DurDuB16}
B.~Durak, T.~DuBuisson, and D.~Cash, ``What else is revealed by order-revealing
  encryption?'' in \emph{Proceedings of the 23rd ACM Conference on Computer and
  Communications Security}, ser. CCS, 2016.

\bibitem{Fit14}
A.~Fitzpatrick, ``Apple says systems weren't hacked in nude pics grab,''
  \url{http://time.com/3257945/apple-icloud-brute-force-jennifer-lawrence/},
  2014.

\bibitem{GarMoh16}
S.~Garg, P.~Mohassel, and C.~Papamanthou, ``Tworam: efficient oblivious {RAM}
  in two rounds with applications to searchable encryption,'' in
  \emph{Proceedings of the 36rd Cryptology Conference}, ser. CRYPTO, 2016.

\bibitem{Gen09}
C.~Gentry, ``Fully homomorphic encryption using ideal lattices,'' in
  \emph{Proceedings of the Symposium on Theory of Computing}, ser. STOC, 2009.

\bibitem{GolOst96}
O.~Goldreich and R.~Ostrovsky, ``Software protection and simulation on
  oblivious {RAMs},'' \emph{Journal of the ACM}, vol.~43, no.~3, 1996.

\bibitem{GruMcP16}
P.~Grubbs, R.~McPherson, M.~Naveed, T.~Ristenpart, and V.~Shmatikov, ``Breaking
  web applications built on top of encrypted data,'' in \emph{Proceedings of
  the 23rd ACM Conference on Computer and Communications Security}, ser. CCS,
  2016.

\bibitem{GruRis17}
P.~Grubbs, T.~Ristenpart, and V.~Shmatikov, ``Why your encrypted database is
  not secure,'' IACR Cryptology ePrint Archive, Tech. Rep. 468, 2017.

\bibitem{GruSek16}
P.~Grubbs, K.~Sekniqi, V.~Bindschaedler, M.~Naveed, and T.~Ristenpart,
  ``Leakage-abuse attacks against order-revealing encryption,'' IACR Cryptology
  ePrint Archive, Tech. Rep. 895, 2016.

\bibitem{HacIye02a}
H.~Hacig{\"u}m{\"u}s, B.~R. Iyer, C.~Li, and S.~Mehrotra, ``Executing sql over
  encrypted data in the database-service-provider model,'' in \emph{Proceedings
  of the ACM International Conference on Management of Data}, ser. SIGMOD,
  2002.

\bibitem{HahKer14}
F.~Hahn and F.~Kerschbaum, ``Searchable encryption with secure and efficient
  updates,'' in \emph{Proceedings of the 21st ACM Conference on Computer and
  Communications Security}, ser. CCS, 2014.

\bibitem{HahKer16}
------, ``Poly-logarithmic range queries on encrypted data with small
  leakage,'' in \emph{Proceedings of the ACM Workshop on Cloud Computing
  Security Workshop}, ser. CCSW, 2016.

\bibitem{HilKos11}
S.~Hildenbrand, D.~Kossmann, T.~Sanamrad, C.~Binnig, F.~F\"{a}rber, and
  J.~W\"{o}hler, ``Query processing on encrypted data in the cloud,''
  Department of Computer Science, ETH Zurich, Tech. Rep. 735, 2011.

\bibitem{HorMeh04}
B.~Hore, S.~Mehrotra, and G.~Tsudik, ``A privacy-preserving index for range
  queries,'' in \emph{Proceedings of the 30th International Conference on Very
  Large Data Bases}, ser. VLDB, 2004.

\bibitem{HwaKim15}
Y.~H. Hwang, S.~Kim, and J.~W. Seo, ``Fast order-preserving encryption from
  uniform distribution sampling,'' in \emph{Proceedings of the ACM Workshop on
  Cloud Computing Security Workshop}, ser. CCSW, 2015.

\bibitem{IslKuz12}
M.~Islam, M.~Kuzu, and M.~Kantarcioglu, ``Access pattern disclosure on
  searchable encryption: ramification, attack and mitigation,'' in
  \emph{Proceedings of the 19th Network and Distributed System Security
  Symposium}, ser. NDSS, 2012.

\bibitem{KadAma10a}
H.~Kadhem, T.~Amagasa, and H.~Kitagawa, ``Mv-opes: multivalued-order preserving
  encryption scheme: a novel scheme for encrypting integer value to many
  different values,'' \emph{IEICE Transactions on Information and Systems},
  vol. E93.D, pp. 2520--2533, 2010.

\bibitem{KadAma10b}
------, ``A secure and efficient order preserving encryption scheme for
  relational databases,'' in \emph{Proceedings of the International Conference
  on Knowledge Management and Information Sharing}, ser. KMIS, 2010.

\bibitem{KamPap13}
S.~Kamara and C.~Papamanthou, ``Parallel and dynamic searchable symmetric
  encryption,'' in \emph{Proceedings of the 17th International Conference on
  Financial Cryptography and Data Security}, ser. FC, 2013.

\bibitem{KamPap12}
S.~Kamara, C.~Papamanthou, and T.~Roeder, ``Dynamic searchable symmetric
  encryption,'' in \emph{Proceedings of the 19th ACM Conference on Computer and
  Communications Security}, ser. CCS, 2012.

\bibitem{KatSah08}
J.~Katz, A.~Sahai, and B.~Waters, ``Predicate encryption supporting
  disjunctions, polynomial equations, and inner products,'' in \emph{Advances
  in Cryptology}, ser. EUROCRYPT, 2008.

\bibitem{KelKol16}
G.~Kellaris, G.~Kollios, K.~Nissim, and A.~O'Neill, ``Generic attacks on secure
  outsourced databases,'' in \emph{Proceedings of the 23rd ACM Conference on
  Computer and Communications Security}, ser. CCS, 2016.

\bibitem{Ker15}
F.~Kerschbaum, ``Frequency-hiding order-preserving encryption,'' in
  \emph{Proceedings of the 22nd ACM Conference on Computer and Communications
  Security}, ser. CCS, 2015.

\bibitem{KerSch14}
F.~Kerschbaum and A.~Schr\"opfer, ``Optimal average-complexity ideal-security
  order-preserving encryption,'' in \emph{Proceedings of the 21st ACM
  Conference on Computer and Communications Security}, ser. CCS, 2014.

\bibitem{KreYak15}
S.~Krendelev, M.~Yakovlev, and M.~Usoltseva, ``Secure database using
  order-preserving encryption scheme based on arithmetic coding and noise
  function,'' in \emph{Proceedings of the 3rd IFIP International Conference on
  Information and Communication Technology}, ser. ICT-EurAsia, 2015.

\bibitem{LacMin17}
M.-S. Lacharité, B.~Minaud, and K.~Paterson, ``Improved reconstruction attacks
  on encrypted data using range query leakage,'' IACR Cryptology ePrint
  Archive, Tech. Rep. 701, 2017.

\bibitem{LacPat15}
M.-S. Lacharité and K.~Paterson, ``A note on the optimality of frequency
  analysis vs.~$\ell_p$-optimization,'' IACR Cryptology ePrint Archive, Tech.
  Rep. 1158, 2015.

\bibitem{LeePar09}
S.~Lee, T.-J. Park, D.~Lee, T.~Nam, and S.~Kim, ``Chaotic order preserving
  encryption for efficient and secure queries on databases,'' \emph{IEICE
  Transactions on Information and Systems}, vol. E92.D, pp. 2207--2217, 2009.

\bibitem{LewWu16}
K.~Lewi and D.~Wu, ``Order-revealing encryption: New constructions,
  applications, and lower bounds,'' in \emph{Proceedings of the 23rd ACM
  Conference on Computer and Communications Security}, ser. CCS, 2016.

\bibitem{LiZha15}
K.~Li, W.~Zhang, C.~Yang, and N.~Yu, ``Security analysis on one-to-many order
  preserving encryption-based cloud data search,'' \emph{IEEE Transactions on
  Information Forensics and Security}, vol.~10, no.~9, 2015.

\bibitem{LiuWan12}
D.~Liu and S.~Wang, ``Programmable order-preserving secure index for encrypted
  database query,'' in \emph{Proceedings of the 5th International Conference on
  Cloud Computing}, ser. CLOUD, 2012.

\bibitem{LiuWan13}
------, ``Nonlinear order preserving index for encrypted database query in
  service cloud environments,'' \emph{Concurrency and Computation: Practice and
  Experience}, vol.~25, no.~13, pp. 1967--1984, 2013.

\bibitem{Lu12}
Y.~Lu, ``Privacy-preserving logarithmic-time search on encrypted data in
  cloud,'' in \emph{Proceedings of the 19th Network and Distributed System
  Security Symposium}, ser. NDSS, 2012.

\bibitem{MavChe15}
C.~Mavroforakis, N.~Chenette, A.~O'Neill, G.~Kollios, and R.~Canetti, ``Modular
  order-preserving encryption, revisited,'' in \emph{Proceedings of the ACM
  International Conference on Management of Data}, ser. SIGMOD, 2015.

\bibitem{McC16}
K.~McCarthy, ``Panama papers hack: unpatched wordpress, drupal bugs to blame?''
  \url{http://www.theregister.co.uk/2016/04/07/panama_papers_unpatched_wordpress_drupal/},
  2016.

\bibitem{Nav15}
M.~Naveed, ``The fallacy of composition of oblivious {RAM} and searchable
  encryption,'' {IACR} Cryptology ePrint Archive, Tech. Rep. 668, 2015.

\bibitem{NavKam15}
M.~Naveed, S.~Kamara, and C.~V. Wright, ``Inference attacks on
  property-preserving encrypted databases,'' in \emph{Proceedings of the 22nd
  ACM Conference on Computer and Communications Security}, ser. CCS, 2015.

\bibitem{NavPra14}
M.~Naveed, M.~Prabhakaran, and C.~Gunter, ``Dynamic searchable encryption via
  blind storage,'' in \emph{Proceedings of the 35th IEEE Symposium on Security
  and Privacy}, ser. S\&P, 2014.

\bibitem{OezSin03}
G.~{\"O}zsoyoglu, D.~A. Singer, and S.~S. Chung, ``Anti-tamper databases:
  querying encrypted databases,'' in \emph{Proceedings of the 17th Conference
  on Data and Application Security}, ser. DBSEC, 2003.

\bibitem{PopLi13}
R.~A. Popa, F.~H. Li, and N.~Zeldovich, ``An ideal-security protocol for
  order-preserving encoding,'' in \emph{34th IEEE Symposium on Security and
  Privacy}, ser. S\&P, 2013.

\bibitem{PouWri16}
D.~Pouliot and C.~Wright, ``The shadow nemesis: Inference attacks on
  efficiently deployable, efficiently searchable encryption,'' in
  \emph{Proceedings of the 23rd ACM Conference on Computer and Communications
  Security}, ser. CCS, 2016.

\bibitem{RocApo16}
D.~Roche, D.~Apon, S.~Choi, and A.~Yerukhimovich, ``Pope: Partial order
  preserving encoding,'' in \emph{Proceedings of the 23rd ACM Conference on
  Computer and Communications Security}, ser. CCS, 2016.

\bibitem{ShiBet07}
E.~Shi, J.~Bethencourt, H.~T.-H. Chan, D.~X. Song, and A.~Perrig,
  ``Multi-dimensional range query over encrypted data,'' in \emph{Proceedings
  of the 2007 Symposium on Security and Privacy}, ser. S\&P, 2007.

\bibitem{SonWag00}
D.~X. Song, D.~Wagner, and A.~Perrig, ``Practical techniques for searches on
  encrypted data,'' in \emph{Proceedings of the 21st IEEE Symposium on Security
  and Privacy}, ser. S\&P, 2000.

\bibitem{StePap14}
E.~Stefanov, C.~Papamanthou, and E.~Shi, ``Practical dynamic searchable
  symmetric encryption with small leakage,'' in \emph{Proceedings of the 21st
  Network and Distributed System Security Symposium}, ser. NDSS, 2014.

\bibitem{TerYun14}
I.~Teranishi, M.~Yung, and T.~Malkin, ``Order-preserving encryption secure
  beyond one-wayness,'' in \emph{Proceedings of the 20th International
  Conference on Advances in Cryptology}, ser. ASIACRYPT, 2014.

\bibitem{WanRav13}
P.~Wang and C.~Ravishankar, ``Secure and efficient range queries on outsourced
  databases using rp-trees,'' in \emph{Proceedings of the 30th IEEE
  International Conference on Data Engineering}, ser. ICDE, 2013.

\bibitem{WozRos13}
S.~Wozniak, M.~Rossberg, S.~Grau, A.~Alshawish, and G.~Schaefer, ``Beyond the
  ideal object: towards disclosure-resilient order-preserving encryption
  schemes,'' in \emph{Proceedings of the ACM Workshop on Cloud Computing
  Security Workshop}, ser. CCSW, 2013.

\bibitem{XiaYen12a}
L.~Xiao and I.-L. Yen, ``A note for the ideal order-preserving encryption
  object and generalized order-preserving encryption,'' IACR Cryptology ePrint
  Archive, Tech. Rep. 350, 2012.

\bibitem{XiaYen12b}
L.~Xiao, I.-L. Yen, and D.~T. Huynh, ``Extending order preserving encryption
  for multi-user systems,'' IACR Cryptology ePrint Archive, Tech. Rep. 192,
  2012.

\bibitem{ZhaKat16}
Y.~Zhang, J.~Katz, and C.~Papamanthou, ``All your queries are belong to us: the
  power of file-injection attacks on searchable encryption,'' in
  \emph{Proceedings of the 25th USENIX Security Symposium}, ser. USENIX
  SECURITY, 2016.

\end{thebibliography}
}

\end{document}